\documentclass{article}

\usepackage{changes}
\definechangesauthor[name=Khiem Pham, color=purple]{KP}
\definechangesauthor[name=Khang Le, color=red]{KL}

\usepackage{microtype}
\usepackage{graphicx}
\usepackage{subfigure}
\usepackage{booktabs} 
\usepackage{amsmath}
\usepackage{amsfonts} 

\usepackage{epsf}
\usepackage{fancyhdr}
\usepackage{graphics}
\usepackage{graphicx}
\usepackage{psfrag}
\usepackage{microtype}
\usepackage{subfigure}

\usepackage{color}
\usepackage{amsthm}
\usepackage{amsfonts}
\usepackage{amsmath}
\usepackage{amssymb,bbm}
\usepackage{caption}
\usepackage{sidecap}
\sidecaptionvpos{figure}{c}

\usepackage{url}
\usepackage{hyperref}

\newcommand{\KL}{\mathbf{KL}}
\newcommand{\diag}{\textnormal{diag}}
\newtheorem{assumption}{Assumption}
\newtheorem{remark}{Remark}

\long\def\comment#1{}
\comment{
\setlength{\topmargin}{0 in}
\setlength{\textwidth}{5.5 in}
\setlength{\textheight}{8 in}
\setlength{\oddsidemargin}{0.5 in}
}
\newtheorem{lemma}{Lemma}
\newtheorem{theorem}{Theorem}

\newtheorem{definition}{Definition}
\newtheorem{corollary}{Corollary}

\newcommand{\Br}{\mathbb{R}}
\newcommand{\one}{\textbf{1}}

\def\aA{\mathbf{a}}
\def\bB{\mathbf{b}}
\def\xX{\mathbf{x}}
\def\yY{\mathbf{y}}

\graphicspath{{./image/}}

\newcommand{\ba}{\begin{array}}
\newcommand{\ea}{\end{array}}

\newcommand{\Rspace}{\mathbb{R}}
\newcommand{\bigO}{\mathcal{O}}
\newcommand{\bigOtil}{\widetilde{\mathcal{O}}}

\newcommand{\mnorm}[1]{\|{#1}\|_{\scriptscriptstyle \infty}}

\usepackage[accepted]{icml2020}
\icmltitlerunning{On Unbalanced Optimal Transport: An Analysis of Sinkhorn Algorithm}

\begin{document}

\twocolumn[
\icmltitle{On Unbalanced Optimal Transport: An Analysis of Sinkhorn Algorithm}

\icmlsetsymbol{equal}{*}

\begin{icmlauthorlist}
\icmlauthor{Khiem Pham}{equal,1}
\icmlauthor{Khang Le}{equal,1}
\icmlauthor{Nhat Ho}{2}
\icmlauthor{Tung Pham}{1,3}
\icmlauthor{Hung Bui}{1}
\end{icmlauthorlist}

\icmlaffiliation{1}{VinAI Research}
\icmlaffiliation{2}{Department of EECS, University of California, Berkeley}
\icmlaffiliation{3}{Faculty of Mathematics, Mechanics and Informatics, Hanoi University of Science, Vietnam National University}

\icmlcorrespondingauthor{Nhat Ho}{minhnhat@berkeley.edu}

\icmlkeywords{Unbalanced Optimal Transport, Sinkhorn algorithm, Complexity analysis}

\vskip 0.3in
]
\printAffiliationsAndNotice{\icmlEqualContribution}

\begin{abstract}
We provide a computational complexity analysis for the Sinkhorn algorithm that solves the entropic regularized Unbalanced Optimal Transport (UOT) problem between two measures of possibly different masses with at most $n$ components.  We show that the complexity of the Sinkhorn algorithm for finding an $\varepsilon$-approximate solution to the UOT problem is of order  $\bigOtil(n^2/ \varepsilon)$. To the best of our knowledge, this complexity is better than the best known complexity upper bound of the Sinkhorn algorithm for solving the Optimal Transport (OT) problem, which is of order $\bigOtil(n^2/\varepsilon^2)$. Our proof technique is based on the geometric convergence rate of the Sinkhorn updates to the optimal dual solution of the entropic regularized UOT problem and scaling properties of the primal solution. It is also different from the proof technique used to establish the complexity of the Sinkhorn algorithm for approximating the OT problem since the UOT solution does not need to meet the marginal constraints of the measures.
\end{abstract}

\section{Introduction} \label{sec:Introduction}
The Optimal Transport (OT) problem has a long history in mathematics and operation research, originally used to find the optimal cost to transport masses from one distribution to the other distribution~\cite{Villani-03}. Over the last decade, OT has emerged as one of the most important tools to solve interesting practical problems in statistics and machine learning~\cite{Ho-ICML-2017, Arjovsky-2017-Wasserstein, Courty-2017-Optimal, Srivastava-2018-Scalable, peyre2019computational}. Recently, the Unbalanced Optimal Transport (UOT) problem between two measures of possibly different masses has been used in several applications in computational biology~\cite{Schiebinger_Optimal_2019}, computational imaging~\cite{lee2019parallel}, deep learning~\cite{Yang_Scalable_2019} and machine learning and statistics~\cite{frogner2015learning, Janati_Wasserstein_2019}. 

The UOT problem is a regularized version of Kantorovich formulation which places penalty functions on the marginal distributions based on some given divergences~\cite{Liero_Optimal_2018}. When the two measures are the probability distributions, the standard OT is a limiting case of the UOT. Under the discrete setting of the OT problem where each probability distribution has at most $n$ components, the OT problem can be recast as a linear programming problem. The benchmark methods for solving the OT problem are interior-point methods of which the most practical complexity is $\bigOtil(n^3)$ developed by~\cite{Pele-2009-Fast}. Recently,~\citep{Lee-2014-Path} used Laplacian linear system algorithms to improve the complexity of interior-point methods to $\bigOtil(n^{5/2})$. However, the interior-point methods are not scalable when $n$ is large. 

In order to deal with the scalability of computing the OT,~\cite{Cuturi-2013-Sinkhorn} proposed to regularize its objective function by the entropy of the transportation plan, which results in the entropic regularized OT. One of the most popular algorithms for solving the entropic regularized OT is the Sinkhorn algorithm~\cite{Sinkhorn-1974-Diagonal}, which was shown by ~\cite{altschuler2017near} to have a complexity of $\bigOtil(n^2/ \varepsilon^3)$ when used to approximate the OT within an $\varepsilon$-accuracy. In the same article,~\cite{altschuler2017near} developed a greedy version of the Sinkhorn algorithm, named the Greenkhorn algorithm, that has a better practical performance than the Sinkhorn algorithm. Later, the complexity of the Greenkhorn algorithm  was improved to $\bigOtil(n^2/ \varepsilon^2)$ by a deeper analysis in ~\cite{Lin-2019-Efficient}. In order to accelerate the Sinkhorn and Greenkhorn algorithms,~\cite{Lin-2019-Acceleration} introduced Randkhorn and Gandkhorn algorithms that have complexity upper bounds of $\mathcal{O}(n^{7/3}/ \varepsilon^{4/3})$. These complexities are better than those of Sinkhorn and Greenkhorn algorithms in terms of the desired accuracy $\varepsilon$. A different line of algorithms for solving the OT problem is based on primal-dual algorithms. These algorithms include accelerated primal-dual gradient descent algorithm~\cite{Dvurechensky-2018-Computational}, accelerated primal-dual mirror descent algorithm~\cite{Lin-2019-Efficient}, and accelerated primal-dual coordinate descent algorithm~\cite{Guo-2019-Accelerated}. These primal-dual algorithms all have complexity upper bounds of $\bigOtil(n^{2.5}/ \varepsilon)$, which are better than those of Sinkhorn and Greenkhorn algorithms in terms of $\varepsilon$. Recently,~\cite{Jambulapati-2019-Direct, Blanchet-2018-Towards} developed algorithms with complexity upper bounds of $\bigOtil(n^{2}/ \varepsilon)$, which are believed to be optimal, based on either a dual extrapolation framework with area-convex mirror mapping or some black-box and specialized graph algorithms. However, these algorithms are quite difficult to implement. Therefore, they are less competitive than Sinkhorn and Greenkhorn algorithms in practice.

\textbf{Our Contribution.} While the complexity theory for OT has been rather well-understood, that for UOT is still nascent.  In the paper, we establish the complexity of approximating UOT between two discrete measures with at most $n$ components. We focus on the setting when the penalty functions are Kullback-Leiber (KL) divergences. Similar to the entropic regularized OT, in order to account for the scalability of computing UOT, we also consider an entropic version of UOT, which we refer to as \emph{entropic regularized UOT}. The Sinkhorn algorithm is widely used to solve the entropic regularized UOT~\cite{Chizat_Scaling_2016}; however, its complexity for approximating the UOT has not been studied. Our contribution is to prove  that the Sinkhorn algorithm has a complexity of 
\begin{align*}
\bigO \biggr( \frac{\tau (\alpha + \beta) n^2}{\varepsilon} \log(n) \biggr[ \log(\|C\|_{\infty}) + \log(\log(n)) & \\
& \hspace{-4 em} + \log \left( \frac{1}{\varepsilon} \right) \biggr] \biggr),
\end{align*}
where $C$ is a given cost matrix, $\alpha, \beta$ are total masses of the measures, and $\tau$ is a regularization parameter with the KL divergences in the UOT problem. This complexity is close to the probably optimal one  by a factor of logarithm of $n$ and $1/\varepsilon$. 

 The main difference between finding an $\varepsilon$-approximation solution  for OT and UOT by the Sinkhorn algorithm  is that the Sinkhorn algorithm for OT knows when it is close to the solution because of the constraints on the marginals, while the UOT does not have that advantage. Despite lacking that useful property, the Sinkhorn algorithm for UOT enjoys more freedom in its updates resulting in some interesting equations that relate the optimal value of the primal objective function of UOT to the masses of two measures (see Lemma~\ref{lemma:key_equation}). Those equations together with the geometric convergence rate of the dual solution of the UOT prove the almost optimal convergence of Sinkhorn algorithm to an $\varepsilon$-approximation solution of the UOT.  

\textbf{Organization.} The remainder of the paper is organized as follows. In Section~\ref{sec:unbalanced_OT}, we provide a setup for the regularized UOT with KL divergences in primal and dual forms, respectively. Based on the dual form, we show that the dual solution has a geometric convergence rate in Section~\ref{sec:complex_UOT}. We also show in Section \ref{sec:complex_UOT} that the Sinkhorn algorithm for the UOT  has a complexity of order  $\tilde{\bigO}(n^2/\varepsilon)$. Section \ref{sec:experiments} presents some empirical results confirming the complexity of the Sinkhorn algorithm. Finally, we conclude in Section \ref{sec:discussions} while deferring the proofs of remaining results in the supplementary material.

\textbf{Notation.} We let $[n]$ stand for the set $\{1, 2, \ldots, n\}$ while $\Rspace^n_+$ stands for the set of all vectors in $\Rspace^n$ with nonnegative components for any $n \geq 2$. For a vector $x \in \Rspace^n$ and $1 \leq p \leq \infty$, we denote $\|x\|_p$ as its $\ell_p$-norm and $\text{diag}(x)$ as the diagonal matrix with $x$ on the diagonal. The natural logarithm of a vector $\aA = (\aA_1,..., \aA_n) \in \mathbb{R}^n$ is denoted $\log \aA = (\log \aA_1,..., \log \aA_n)$. $\one_n$ stands for a vector of length $n$ with all of its components equal to $1$. $\partial_x f$ refers to a partial gradient of $f$ with respect to $x$. Lastly, given the dimension $n$ and accuracy $\varepsilon$, the notation $a = \bigO\left(b(n,\varepsilon)\right)$ stands for the upper bound $a \leq C \cdot b(n, \varepsilon)$ where $C$ is independent of $n$ and $\varepsilon$. Similarly, the notation $a = \bigOtil(b(n, \varepsilon))$ indicates the previous inequality may depend on the logarithmic function of $n$ and $\varepsilon$, and where $C>0$. 
\section{Unbalanced Optimal Transport with Entropic Regularization}
\label{sec:unbalanced_OT}
In this section, we present the primal and dual forms of the entropic regularized UOT problem and define an $\varepsilon$-approximation for the solution of the unregularized UOT. 

For any two positive vectors $\aA = (\aA_1,\ldots,\aA_n) \in \Br_{+}^{n}$ and $\bB  = (\bB_1,\ldots,\bB_n)\in \Br_{+}^{n}$, the UOT problem takes the form  $ \min_{X\in \Br_{+}^{n \times n}} f(X)  $ where
\begin{align}
    f(X) : = \left\langle C, X\right\rangle + \tau \KL(X \one_{n} || \aA) & \label{eq:unbalanced_OT} \\
    & \hspace{-0.5 em} + \tau \KL(X^{\top} \one_{n} || \bB). \nonumber
\end{align}
Here, $C$ is a given cost matrix, $X$ is a transportation plan, $\tau > 0$ is a given regularization parameter, and the $\KL$ divergence between vectors $\xX$ and $\yY$ is defined as 
 \begin{align*}
 \KL(\xX\|\yY) : =  \sum_{i=1}^n \xX_i \log \left( \frac{\xX_i}{\yY_i} \right) - \xX_i + \yY_i.
 \end{align*}
When $\aA^{\top} \one_{n} = \bB^{\top} \one_{n}$ and $\tau \to \infty$, the UOT problem becomes the standard OT problem. Similar to the original OT problem, the exact computation of UOT is expensive and not scalable in terms of the number of supports $n$. Inspired by the recent success of the entropic regularized OT problem as an efficient approximation of the OT problem, we also consider the entropic version of the UOT problem~\cite{frogner2015learning} , which we refer to as entropic regularized UOT, of finding $\min_{X\in \Br_{+}^{n \times n}} g(X)$, where   
\begin{align}
     g(X) : =  \left\langle C, X\right\rangle - \eta H(X) + \tau \KL(X \one_{n} || \aA) \nonumber\\
      + \tau \KL(X^{\top} \one_{n} || \bB). \label{eq:entro_unbalanced_OT}
\end{align}
Here, $\eta > 0$ is a given regularization parameter and $H(X)$ is an entropic regularization given by
\begin{align}
    H(X) : = - \sum_{i, j=1}^n X_{ij} ( \log ( X_{ij}) - 1). \label{eq:entro_formu}
\end{align}
For each $\eta > 0$, we can check that the entropic regularized UOT problem is strongly convex in $X$. Therefore, it is convenient to solve for the optimal solution of the entropic regularized UOT and uses it to approximate the original value of UOT. 
\begin{definition} \label{def:approx}
For any $\varepsilon > 0$, we call $X$ an $\varepsilon$-approximation transportation plan if the following holds
\begin{align*}
& \left\langle C, X\right\rangle + \tau \KL(X \one_{n} || \aA) + \tau \KL(X^{\top} \one_{n} || \bB) \\
& \leq \left\langle C, \widehat{X}\right\rangle + \tau \KL(\widehat{X} \one_{n} || \aA) + \KL(\widehat{X}^{\top} \one_{n} || \bB) + \varepsilon,
\end{align*}
where $\widehat{X}$ is an optimal transportation plan for the UOT problem~\eqref{eq:unbalanced_OT}.
\end{definition}
We aim to develop an algorithm to obtain $\varepsilon$-approximation transportation plan for the UOT problem~\eqref{eq:unbalanced_OT}. In order to do that, we consider the Fenchel-Legendre dual form of the entropic regularized UOT, which is given by
\begin{align*}
    \resizebox{\hsize}{!}{$\displaystyle \max_{u, v \in \Br^{n}} - F^{*}(- u) - G^{*}(- v) - \eta \sum_{i, j} \exp \left( \frac{u_{i} + v_{j} - C_{ij}}{\eta} \right)$,}
\end{align*}
where the functions $F^{*}(.)$ and $G^{*}(.)$ take the following forms:
{\footnotesize
\begin{align*}
    F^{*}(u) &= \displaystyle \max_{z \in \Br^{n}} z^{\top} u - \tau \KL(z||\aA) = \tau \left\langle e^{u/ \tau}, \aA \right\rangle - \aA^{\top}\one_{n}, \\
    G^{*}(v) &= \displaystyle \max_{x \in \Br^{n}} x^{\top} v - \tau \KL(x||\bB) = \tau \left\langle e^{v/ \tau}, \bB \right\rangle - \bB^{\top}\one_{n}.
\end{align*}
\par}
Since $\aA$ and $\bB$ are given non-negative vectors, finding the optimal solution for the above objective is equivalent to finding the optimal solution for the following objective
{\footnotesize
\begin{align}
\min_{u, v \in \Br^{n}} h(u, v) &:= \eta \sum_{i, j = 1}^n \exp \left( \frac{u_{i} + v_{j} - C_{ij}}{\eta} \right) \nonumber \\
& \quad + \tau \left\langle e^{- u/ \tau}, \aA \right\rangle + \tau \left\langle e^{- v/ \tau}, \bB \right\rangle. 
\label{eq:dual_unbalanced_OT} 
\end{align}
\par}
Problem~\eqref{eq:dual_unbalanced_OT} is referred to as the \textit{dual entropic regularized UOT}. 
\section{Complexity Analysis of Approximating Unbalanced Optimal Transport}
\label{sec:complex_UOT}
In this section, we provide a complexity analysis of the Sinkhorn algorithm for approximating UOT solution. We start with some notations and useful quantities followed by the lemmas and main theorems.
\subsection{Notations and Assumptions}
We first denote $\sum_{i=1}^n \aA_i = \alpha$, $\sum_{j=1}^n \bB_j = \beta$. For each $u,v\in \mathbb{R}^{n}$, its corresponding optimal transport in the dual form \eqref{eq:dual_unbalanced_OT} is denoted by  $B(u,v)$, where $B(u,v) : = \diag(e^{u/ \eta}) \ e^{-\frac{C}{\eta}} \ \diag(e^{v/ \eta})$. 
The corresponding solution in \eqref{eq:entro_unbalanced_OT} is denoted by $X = B(u,v)$. Let $a = B(u,v) \one_n$, $b = B(u,v)^{\top} \one_n$ and $\sum_{i,j=1}^n X_{ij} = x$.\\
Let $(u^k, v^k)$ be the solution returned at the $k$-th iteration of the Sinkhorn algorithm and $(u^*,v^*)$ be the optimal solution of \eqref{eq:dual_unbalanced_OT}. Following the above scheme, we also define $X^k, a^k, b^k$, $x^k$ and $X^*, a^*, b^*$, $x^*$ correspondingly. Additionally, we define $\widehat{X}$ to be the optimal solution of the unregularized objective \eqref{eq:unbalanced_OT} and $\sum_{i,j=1}^n \widehat{X}_{ij} = \widehat{x}$.

Different from the standard OT problem, the optimal solutions of the entropic regularized UOT and our complexity analysis overall also depend on the masses $\alpha, \beta$ and the $\KL$ regularization parameter $\tau$. Given that, we will assume the following simple and standard regularity conditions throughout the paper.

\textbf{Regularity Conditions:}
\begin{itemize}
\item[(A1)] $\alpha$,  $\beta, \tau $ are positive constants.
\item[(A2)] $C$ is a matrix of non-negative entries.
\end{itemize}

Before presenting the main theorem and analysis, we define some quantities that will be used in our analysis and quantify their magnitudes under the above regularity conditions.

\textbf{List of Quantities:} 
{\footnotesize
\begin{align}
    R &= \max\Big\{ \mnorm{\log(\aA)} , \mnorm{\log(\bB)} \Big\} \nonumber\\
    & \quad + \max \left\{ \log(n),  \frac{1}{\eta}\mnorm{C} - \log(n) \right\}, \label{quan:R}\\
    \Delta_k &=  \max\Big\{\|u^k-u^*\|_{\infty},\|v^k - v^*\|_{\infty} \Big\}, \label{quan:Delta_k} \\
    \Lambda_k &= \tau \Big(\frac{\tau}{\tau + \eta} \Big)^{k} R . \label{quan:Lambda_k}\\
    S & = \frac{1}{2}(\alpha+\beta) + \frac{1}{2} + \frac{1}{4\log(n)}, \label{quan:S}\\
    T &= \left( \frac{\alpha+\beta}{2} \right) \left[\log \left( \frac{\alpha+\beta}{2} \right) \nonumber + 2\log(n)-1\right] \\
    & \quad + \log(n) + \frac{5}{2}, \label{quan:T} \\
    U&= \max \Big\{S + T , 2\varepsilon , \frac{4\varepsilon \log(n)}{\tau}, \frac{4\varepsilon (\alpha+\beta) \log(n)}{\tau}\Big\}. \label{quan:U}
\end{align}
\par}
As we shall see, the quantities $\Delta_k$ and $\Lambda_k$ are used to establish the convergence rate of $(u_k, v_k)$. We now consider the order of $R, S$ and $T$. Since the order of the penalty function $\eta H(X)$ is $\bigO(\eta\log(n))$ and should be small for a good approximation, $\eta$ is often chosen such that  $\eta \log(n)$ is sufficiently small (see our choice of $\eta = \frac{\varepsilon}{U}$ in Theorem \ref{theorem:Sinkhorn_to_optimal}). Therefore, we can assume  the dominant factor in the second term of $R$ is  $\frac{1}{\eta} \|C\|_{\infty}$. If $\alpha = \sum_{i=1}^n \aA_i$ is a positive constant, then we can expect that $\aA_i$  is as small as $\bigO(n^{-\kappa})$ for a constant $\kappa \geq 1$. In this  case,  $\|\log(\aA)\|_{\infty} = \bigO(\log(n))$. Overall, we can assume that $R = \bigO\big(\frac{1}{\eta} \|C\|_{\infty}\big)$ and if $\alpha, \beta$ and $\tau$ are positive constants, then $S= \bigO\big(1 \big)$ and $T = \bigO\big(\log(n)\big)$.
 
\subsection{Sinkhorn Algorithm}
\label{subsec:Sinkhorn}
The Sinkhorn algorithm \cite{Chizat_Scaling_2016} alternatively minimizes the dual function in \eqref{eq:dual_unbalanced_OT} with respect to $u$ and $v$. Suppose we are at iteration $k+1$ for $k\ge 0$ and $k$ even, by setting the gradient to $0$ we can see that given fixed $v^k$, the update $u^{k+1}$ that minimizes the function in \eqref{eq:dual_unbalanced_OT} satisfies
\begin{align*}
\exp \left( \frac{u^{k+1}_{i}}{\eta} \right) \sum_{j = 1}^{n} \exp \left( \frac{v^{k}_{j} - C_{ij}}{\eta} \right) = \exp \left(-\frac{u^{k+1}_{i}}{\tau} \right) \aA_{i}.
\end{align*}
 Multiplying both sides by $\exp\left(\frac{u_i^k}{\eta}\right)$, we get:
\begin{align*}
\exp \biggr( \frac{u^{k+1}_{i}}{\eta} \biggr) a_i^k = \exp\left(\frac{u_i^k}{\eta}\right) \exp \left( -\frac{u^{k+1}_{i}}{\tau} \right) \aA_{i}.
\end{align*}
Similarly with $u^k$ fixed and $k$ is odd:
\begin{equation*}
\exp \biggr( \frac{v^{k+1}_{j}}{\eta} \biggr) b_j^k = \exp\left(\frac{v_j^k}{\eta}\right) \exp \left( -\frac{v^{k+1}_{j}}{\tau} \right) \bB_{j}.
\end{equation*}
These equations translate to the pseudocode of Algorithm ~\ref{Algorithm:Sinkhorn}, in which we have included our stopping condition stated in Theorem~\ref{theorem:Sinkhorn_to_optimal}.
\begin{algorithm}[t]
\caption{UNBALANCED\_SINKHORN$\left(C, \varepsilon\right)$} \label{Algorithm:Sinkhorn}
\begin{algorithmic}
\STATE \textbf{Input:} $k = 0$ and $u^0 = v^0 = 0$ and  $\eta = \varepsilon/U$ 
\WHILE{\resizebox{0.78\hsize}{!}{$k \le \Big(\frac{\tau U}{\varepsilon} + 1\Big) \Big[\log(8\eta R \big) + \log(\tau(\tau+1)) + 3\log(\frac{U}{\varepsilon}) \Big]$}}
\STATE $a^{k} = B(u^k, v^k) \one_n$. 
\STATE $b^{k} = B(u^k, v^k)^\top \one_n$. 
\IF{$k$ is even}
\STATE $u^{k+1} = \biggr[ \dfrac{u^k}{\eta} + \log\left(\aA\right) - \log\left(a^{k} \right)\biggr] \dfrac{ \eta \tau}{\eta + \tau}$
\STATE $v^{k + 1} = v^{k}$
\ELSE
\STATE $v^{k+1} = \biggr[ \dfrac{v^k}{\eta} + \log\left(\bB \right) - \log\left(b^{k} \right)\biggr] \dfrac{ \eta \tau}{\eta + \tau}$
\STATE $u^{k + 1} = u^{k}$.
\ENDIF
\STATE $k = k + 1$. 
\ENDWHILE
\STATE \textbf{Output:} $B(u^k, v^k)$.  
\end{algorithmic}
\end{algorithm} 
We now present the main theorems.
\begin{theorem}
\label{theorem:bound_optimal_solution}
For any $k \geq 0$, the update $(u^{k+1}, v^{k+1})$ from Algorithm~\ref{Algorithm:Sinkhorn} satisfies the following bound
\begin{equation}
\begin{split}
\Delta_{k + 1} \leq \Lambda_k,
\end{split}
\label{eq:bound_optimal_solution}
\end{equation}
where $\Delta_k$ and  $\Lambda_k$ are defined in \eqref{quan:Delta_k} and \eqref{quan:Lambda_k}, respectively.
\end{theorem}
\begin{remark} Theorem \ref{theorem:bound_optimal_solution} establishes a geometric convergence rate for the dual solution $(u^k,v^k)$ under $\ell_{\infty}$ norm. Its geometric convergence is similar to the work of~\cite{sejourne2019UOT} while it is different from the work of~\cite{Chizat_Scaling_2016}, which used the Thompson metric. The difference between the result of Theorem 1 and that in~\cite{sejourne2019UOT} is that we obtain a specific upper bound for the convergence rate of the Sinkhorn updates, which depends explicitly on the number of components $n$ and all other parameters of masses and penalty functions. The convergence rate of Theorem~\ref{theorem:bound_optimal_solution} plays an important role in the complexity analysis of the Sinkhorn algorithm in the next theorem.
\end{remark}

\begin{theorem}\label{theorem:Sinkhorn_to_optimal}
Under the regularity conditions (A1-A2), with $R$ and $U$ defined in \eqref{quan:R} and \eqref{quan:U} respectively, for $\eta = \frac{\varepsilon}{U}$ and
\begin{align*}
k \ge 1 +  \left(\frac{\tau U}{\varepsilon} + 1 \right) \biggr[ \log(8\eta R \big) + \log(\tau(\tau+1)) & \\
& \hspace{- 3 em} + 3\log (\frac{U}{\varepsilon}) \biggr],
\end{align*} 
the update $X^k$ from Algorithm~\ref{Algorithm:Sinkhorn} is an $\varepsilon$-approximation of the optimal solution $\widehat{X}$ of \eqref{eq:unbalanced_OT}.
\end{theorem}
 
The next corollary sums up the complexity of Algorithm \ref{Algorithm:Sinkhorn}.
\begin{corollary}\label{corollary:complexity}
Under conditions (A1-A2) and assume that $R = \bigO \left(\frac{1}{\eta}\|C\|_{\infty} \right), S = \bigO(1)$ and $T = \bigO(\log(n))$. Then the complexity of Algorithm \ref{Algorithm:Sinkhorn}  is 
\begin{align*}
\bigO \biggr(\dfrac{\tau (\alpha + \beta) n^2}{\varepsilon} \log(n) \biggr[ \log(\mnorm{C}) + \log(\log(n)) & \\
& \hspace{- 4 em} + \log \left( \frac{1}{\varepsilon} \right) \biggr] \biggr).
\end{align*}
\end{corollary}

\begin{proof}[Proof of Corollary \ref{corollary:complexity}]
By the assumptions on the order of $R, S$, $T$, we have
\begin{align*}
    R = \bigO(\frac{1}{\eta}\mnorm{C}), S = \bigO(\alpha + \beta), T = \bigO((\alpha+\beta)\log(n)).
\end{align*}
Plugging the above results into the definition of $U$ in \eqref{quan:U}, we find that
\begin{align*}
     U = \bigO(S) + \bigO(T) + \varepsilon \bigO(\log(n)) = \bigO((\alpha+\beta)\log(n)).
\end{align*}
Overall, we obtain
\begin{align*}
    k = \bigO \biggr( \frac{\tau (\alpha + \beta) \log(n)}{\varepsilon} \biggr) \biggr[\log(\|C\|_{\infty}) + \log(\log(n)) & \\
    & \hspace{-5 em} + \log \left( \frac{1}{\varepsilon} \right)  \biggr] \biggr).
\end{align*}
By multiplying the above bound of $k$ with $\bigO(n^2)$ arithmetic operations per iteration, we obtain the desired final complexity of the Sinkhorn algorithm as being claimed in the conclusion of the corollary.
\end{proof}
\begin{remark} 
Since $\alpha$, $\beta$ and $\tau$ are positive constants, we obtain the complexity $\bigOtil\left(n^2/\epsilon\right)$ as stated in the abstract. In comparison to the best known OT's complexity of the similar order of $n$, i.e.,~\cite{Dvurechensky-2018-Computational, Lin-2019-Efficient}, our complexity for the UOT is better by a factor of $\varepsilon$. Meanwhile, among the practical algorithms for  OT which have complexity of the order of $n^{7/3}$, i.e., Gankhorn and Randkhorn algorithms~\cite{Lin-2019-Acceleration}, our bound is better by a factor of $n^{1/3}$. 
\end{remark} 

\subsection{Proof of Theorem~\ref{theorem:bound_optimal_solution}}
 The analysis of Sinkhorn algorithm for approximating unbalanced optimal transport is different from that of optimal transport, since $\aA$ and $\bB$ need not be probability measures. Our proof of Theorem \ref{theorem:bound_optimal_solution} requires the geometric convergence rate of $(u^k, v^k)$ and an upper bound on the supremum norm of the optimal dual solution $(u^*, v^*)$. The latter result is presented in Lemma \ref{lemma:sup_norm}. Before stating that result formally, we need the following lemmas.
\begin{lemma}
\label{lemma:optimal_solution}
The optimal solution $(u^*, v^*)$ of dual entropic regularized UOT~\eqref{eq:dual_unbalanced_OT} satisfies the following equations:
\begin{equation*}
\frac{u^*}{\tau} = \log(\aA) - \log(a^*), \ \text{and} \ \frac{v^*}{\tau} = \log(\bB) - \log(b^*). 
\end{equation*}
\end{lemma}
\begin{proof}
Since $(u^*,v ^*)$ is a fixed point of the updates in the Algorithm~\ref{Algorithm:Sinkhorn}, we get
\begin{align*}
    u^* &= \left[\frac{u^*}{\eta} + \log(\aA) - \log(a^*) \right] \frac{\eta \tau}{\eta + \tau}.
\end{align*}
This directly leads to the stated equality for $u^*$, and that for $v^*$ can be obtained similarly. 
\end{proof}

\begin{lemma}
\label{lemma:kth_inequality}
Assume that the regularity conditions (A1-A2) hold. Then, the following are true:
\begin{align*}
\begin{split}
 &(a)\quad  \Big| \log\left( \frac{a_i^*}{a_i^k} \right) -  \frac{u_i^* - u_i^k}{\eta}\Big| \leq  \max_{1 \leq j \leq n} \frac{|v^*_j - v^{k}_j|}{\eta},\\
 & (b)\quad  \Big| \log\left( \frac{b_j^*}{b_j^k} \right) -  \frac{v_j^* - v_j^k}{\eta}\Big| \leq  \max_{1 \leq i \leq n} \frac{|u^*_i - u^{k}_i|}{\eta}.
\end{split}
\end{align*}
\end{lemma}
The proof is given in the supplementary material.
\begin{lemma}
\label{lemma:sup_norm}
The sup norms of the optimal solution $\mnorm{u^*}$ and $\mnorm{v^*}$ are bounded by:
\begin{equation*}
\max \{ \mnorm{u^*},\mnorm{v^*} \} \le \tau R,
\end{equation*}
where $R$ is defined in equation~\eqref{quan:R}.
\end{lemma}
\begin{proof}
We start with the equations for the solution $u^*$ in Lemma \ref{lemma:optimal_solution}, i.e., we have
\begin{align*}
&\frac{u_i^*}{\tau} = \log(\aA_i) - \log \left(\sum_{j=1}^n \exp\left(\frac{u_i^* + v_j^* - C_{ij}}{\eta} \right)  \right),
\end{align*}
which can be rewritten as
\begin{align*}
&u_i^* \left(\frac{1}{\tau} + \frac{1}{\eta} \right) = \log (\aA_i) - \log \left[\sum_{j=1}^n \exp\left(\frac{v_j^*  - C_{ij}}{\eta} \right) \right].
\end{align*}

The second term in the above display can be lower bounded as follows
{\footnotesize
\begin{align*}
\log\left[\sum_{j=1}^n \exp\left(\frac{v_j^* - C_{ij}}{\eta} \right)\right]
     &\ge \log(n) + \min_{1 \le j \le n} \left\{ \frac{v_j^* - C_{ij}}{\eta} \right\} \\
     &\ge \log(n) - \frac{\mnorm{v^*}}{\eta} - \frac{\mnorm{C}}{\eta}. 
\end{align*}
\par}
Additionally, we also have
{\footnotesize
\begin{align*}
\log\left[\sum_{j=1}^n \exp\left(\frac{v_j^* - C_{ij}}{\eta} \right)\right] 
&\le \log(n) + \max_{1 \le j \le n} \left\{ \frac{v_j^* - C_{ij}}{\eta} \right\} \\
&\le \log(n) + \frac{\mnorm{v^*}}{\eta}.
\end{align*}
\par}
Collecting the above results, we find that
\begin{align*}
     &\left|\log\left[\sum_{j=1}^n \exp\left(\frac{v_j^* - C_{ij}}{\eta} \right)\right] \right| \\
     &\le \frac{\mnorm{v^*}}{\eta} + \max \left\{ \log(n), \frac{\mnorm{C^*}}{\eta} - \log(n) \right\}.
\end{align*}


Hence, we obtain that
\begin{align*}
   |u_i^*| \left(\frac{1}{\tau} + \frac{1}{\eta} \right) &\le |\log (\aA_i)| + \frac{\mnorm{v^*}}{\eta} + \\
   &\qquad \max \left\{ \log(n), \frac{\mnorm{C^*}}{\eta} - \log(n) \right\}.
\end{align*}
By choosing an index $i$ such that $|u_i^*| = \mnorm{u^*}$ and combining with the fact that $|\log (\aA_i)| \le \max \{ \mnorm{\log(\aA)}, \mnorm{\log(\bB)} \}$, we have
\begin{align*}
   \mnorm{u^*} \left(\frac{1}{\tau} + \frac{1}{\eta} \right) \le \frac{\mnorm{v^*}}{\eta} + R.
\end{align*}
WLOG assume that $\mnorm{u^*} \ge \mnorm{v^*}$. Then, we obtain the stated bound in the conclusion.
\end{proof}

\paragraph{Proof of Theorem~\ref{theorem:bound_optimal_solution}}
We first  consider the case when $k$ is even. From the update of $u^{k + 1}$  in Algorithm~\ref{Algorithm:Sinkhorn}, we have:
{\footnotesize
\begin{align*}
&u^{k+1}_i = \Big(\frac{u_i^k}{\eta} + \log \aA_i - \log a^{k}_i \Big) \frac{\eta \tau}{\tau + \eta} \\
&\space \space = \left\{ \frac{u_i^k}{\eta} + \left[ \log (\aA_i) - \log (a^*_i) \right] + \left[ \log (a^*_i) - \log (a^{k}_i) \right] \right\} \frac{\eta \tau}{\tau + \eta}.
\end{align*}
\par}
Using Lemma \ref{lemma:optimal_solution}, the above equality is equivalent to
\begin{align*}
    u_i^{k+1} - u_i^* = \left[\eta \log\Big(\frac{a_i^*}{a_i^k}\Big) - (u_i^* - u_i^k) \right] \frac{\tau}{\tau + \eta}.
\end{align*}
Using Lemma \ref{lemma:kth_inequality}, we get
\begin{align*}
    \big|u_i^{k+1} - u_i^*\big| \leq \max_{1\leq j \leq n} \big|v_j^k - v_j^* \big| \frac{\tau}{\tau + \eta}.
\end{align*}

This leads to $\mnorm{u^{k+1} - u^*} \leq  \frac{\tau}{\tau + \eta} \mnorm{v^k - v^*}$. Similarly, we obtain $\mnorm{v^{k} - v^*} \leq  \frac{\tau}{\tau + \eta} \mnorm{u^{k - 1} - u^*}$. Combining the two inequalities yields
\begin{align*}
    \mnorm{u^{k+1} - u^*} \leq \ \Big( \frac{\tau}{\tau + \eta} \Big)^2 \mnorm{u^{k-1} - u^*}.
\end{align*}

Repeating all the above arguments alternatively, we have
\begin{align*}
    \mnorm{u^{k+1} - u^*} & \leq \Big( \frac{\tau}{\tau + \eta} \Big)^{k+1} \mnorm{v^0 - v^*} \\
    & = \Big( \frac{\tau}{\tau + \eta} \Big)^{k+1} \mnorm{v^*}.
\end{align*}
Note that $v^{k + 1} = v^k$ for $k$ even. Therefore, we find that
\begin{align*}
    \mnorm{v^{k+1} - v^*} & \le \frac{\tau}{\tau + \eta} \mnorm{u^{k - 1} - u^*} \\
    & \leq \Big( \frac{\tau}{\tau + \eta} \Big)^{k} \mnorm{v^*}.
\end{align*}
These two results lead to $\Delta_{k + 1} \leq \Big( \frac{\tau}{\tau + \eta} \Big)^{k} \|v^* \|_\infty$. Similarly, for $k$ odd we obtain $\Delta_{k + 1} \leq \Big( \frac{\tau}{\tau + \eta} \Big)^{k} \|v^* \|_\infty$. Thus the above inequality is true for all $k$. Using the fact that $\mnorm{v^*} \le \max\{\mnorm{u^*},\mnorm{v^*}\}$ and Lemma \ref{lemma:sup_norm}, we obtain the conclusion of Theorem~\ref{theorem:bound_optimal_solution}.
\subsection{Proof of Theorem~\ref{theorem:Sinkhorn_to_optimal}}
The proof is based on the upper bound for the convergence rate in Theorem \ref{theorem:bound_optimal_solution} and an upper bound for the solutions  $\widehat{x}$ and $x^*$  of functions in equations~\eqref{eq:unbalanced_OT} and \eqref{eq:entro_unbalanced_OT}, respectively, which are direct consequences of the following lemmas.

\begin{lemma}
\label{lemma:key_equation}
Assume that the function $g$ attains its minimum at $X^*$, then
\begin{equation}
g(X^*) + (2\tau + \eta) x^* = \tau (\alpha + \beta).
\label{eq:g_solution}
\end{equation}
Similarly, assume that $f$ attains its minimum at $\widehat{X}$, then 
\begin{equation}
f(\widehat{X}) + 2\tau\widehat{x} = \tau (\alpha + \beta),
\label{eq:f_solution}
\end{equation}
where $\sum_{i,j=1}^n X^* = x^*$ and $\sum_{i,j=1}^n \widehat{X} = \widehat{x}$.
\end{lemma}
Both equations in Lemma  \ref{lemma:key_equation} establish the relationships between the optimal solutions of functions in equations~\eqref{eq:unbalanced_OT} and \eqref{eq:entro_unbalanced_OT} with other parameters. Those relationships are very useful for analysing the behaviour of the optimal solution of UOT, because  the UOT does not have any conditions on the marginals as the OT does. Consequences of Lemma \ref{lemma:key_equation} include Corollary \ref{corollary:solution_sum_bound} which provides the upper bounds for $\widehat{x}$ and $x^*$ as well as bounds for the entropic functions in the proof of Theorem \ref{theorem:Sinkhorn_to_optimal}. The key idea of the proof surprisingly comes from the fact that the UOT solution does not have to meet the marginal constraints. We note that equation \eqref{eq:g_solution} could also be proved by using the fixed point equations in Lemma \ref{lemma:optimal_solution} as in \cite{janati2019spatio}. Here we offer an alternative proof that can be applied to the unregularized case as well. We now present the proof of Lemma~\ref{lemma:key_equation}.
\begin{proof}
Consider the function $g(tX^*)$, where $t\in \mathbb{R}^+$,
\begin{align*}
g(tX^*) &= \langle C, tX^*\rangle + \tau \text{KL}(tX^* \one_n \| \aA)\\
& \quad + \tau \text{KL}((tX^*)^{\top}\one_n \| \bB) - \eta H(tX^*).
\end{align*}
For the KL term of $g(tX^*)$, we have:
\begin{align*}
&\text{KL}(tX^*\mathbf{1}_n\|\aA) \\
&= \sum_{i=1}^n (t a_i^*) \log\left(\frac{t a_i^*}{\aA_i} \right) - \sum_{i=1}^n (t a_i^*) + \sum_{i=1}^n \aA_i \\
&= \sum_{i=1}^n (t a_i^*) \left(\log\left(\frac{a_i^*}{\aA_i} \right) + \log(t) \right) - t x^* + \alpha \\
& = t \sum_{i=1}^n \left( a_i^*  \log\left(\frac{a_i^*}{\aA_i} \right) - x^* + \aA_i \right) + (1-t)\alpha \\
& \hspace{17 em} + x^*t\log(t) \\
& = t \text{KL}(X^* \mathbf{1}_n \|\aA) + (1-t)\alpha + x^* t\log(t).
\end{align*}
Similarly, we get
\begin{align*}
 \text{KL}(t(X^*)^{\top} \mathbf{1}_n \|\bB) &= t \text{KL}\big((X^*)^{\top}\mathbf{1}_n \|\bB\big) \\
& \quad + (1-t)\beta + x^* t\log(t).
\end{align*}
For the entropic penalty term, we find that
\begin{align*}
    - H(tX^*) &= \sum_{i,j=1}^n tX^*_{ij} \big(\log(tX_{ij}^*) -1\big) \\
    &= \sum_{i,j} t X_{ij}^* \big(\log(X_{ij}^*)-1\big) + x^* t \log(t) \\
    &= -t H(X^*) + x^* t\log(t).
\end{align*}
Putting all results together, we obtain
\begin{align*}
g(tX^*) = t g(X^*) + \tau(1-t)(\alpha+\beta) & \\
& \hspace{-2 em} +  (2\tau  + \eta)x^* t\log(t).
\end{align*}
Taking the derivative of $g(tX^*)$ with respect to $t$,
\begin{equation*}
    \resizebox{\hsize}{!}{$\frac{\partial g(tX^*)}{\partial t} = g(X^*) - \tau(\alpha+\beta) + (2\tau + \eta)x^*(1+\log(t))$.}
\end{equation*}
The function $g(tX^*)$ is well-defined for all $t\in \mathbb{R}^+$.
We know that $g(tX^*)$ attains its minimum at $t=1$. Replace $t=1$ into the above equation, we obtain
\begin{align*}
     g(X^*) - \tau(\alpha+\beta) +  (2\tau + \eta)x^* &= 0 \\
    g(X^*) + (2\tau +\eta) x^* &= \tau(\alpha+\beta).
\end{align*}
The second claim is proved in the same way.
\end{proof}

\begin{corollary}
\label{corollary:solution_sum_bound}
Assume that conditions (A1-A2) hold and $\eta \log(n)$ is sufficiently small. We have the following bounds on $x^*$ and $\widehat{x}$:
\begin{align*}
&(a)\quad x^* \leq \left(\frac{1}{2}+ \frac{\eta \log(n)}{2\tau - 2 \eta \log(n)}\right)(\alpha+\beta) + \frac{1}{6\log(n)},\\
&(b)\quad  \widehat{x} \leq \frac{\alpha + \beta}{2}.
\end{align*}
\end{corollary}
We defer the proof of Corollary \ref{corollary:solution_sum_bound} to the appendix. Next, we use the condition for  $k$ in Theorem \ref{theorem:Sinkhorn_to_optimal} to bound some relevant quantities at the $k$-th iteration of  the Sinkhorn algorithm. 
\begin{lemma}
\label{lemma:bound_xk}
Assume that the regularity conditions (A1-A2) hold and $k$ satisfies the inequality in Theorem \ref{theorem:Sinkhorn_to_optimal}. The following are true
\begin{align*}
&(a) \quad \Lambda_{k - 1} \leq \frac{\eta^2}{8(\tau+1)}, \\
&(b) \quad  |x^k - x^*| \le \frac{3}{\eta} \Delta_k \min\big\{x^*,x^k \big\}, \\
&(c) \quad x^k \leq S,
\end{align*}
where $S$ is defined in equation~\eqref{quan:S}.
\end{lemma}

We are now ready to construct a proof for Theorem \ref{theorem:Sinkhorn_to_optimal}.
\paragraph{Proof of Theorem \ref{theorem:Sinkhorn_to_optimal}}
From the definitions of $f$ and $g$, we have
{\footnotesize
\begin{align}
&f(X^k) - f(\widehat{X}) \nonumber \\
&= g(X^k)+\eta H(X^k) - g(\widehat{X}) - \eta H(\widehat{X}) \nonumber \\
&= g(X^k)+\eta H(X^k) - g(\widehat{X}) - \eta H(\widehat{X}) - g(X^*) + g(X^*) \nonumber \nonumber \\
&\leq \Big[g(X^k)-g(X^*)\Big] + \eta \Big[H(X^k) - H(\widehat{X})\Big],
\label{bound:fx}
\end{align}
\par}
since $g(X^*) - g(\widehat{X})  \le 0$, as $X^*$ is the optimal solution of function~\eqref{eq:entro_unbalanced_OT}. The above two terms can be bounded separately as follows:
\paragraph{Upper Bound of $H(X^k) - H(\widehat{X})$.}  $ $\\
We first show the following inequalities
\begin{equation}
x - x\log(x) \le  H(X) \le  2x\log(n) + x - x\log (x)
\label{ineq:bound_Hx}
\end{equation}
for any $X$ that $X_{ij} \ge 0$ and $x = \sum_{ij} X_{ij}$.

Indeed, rewriting $-H(X)$ as
\begin{align*}
-H(X) = x\Big[\sum_{i,j=1}^n \frac{X_{ij}}{x} \log\Big(\frac{X_{ij}}{x} \Big) -1 \Big]+ x \log(x),
\end{align*}
and using $-2 \log(n) \leq  \sum_{i,j=1}^n \frac{X_{ij}}{x} \log\Big(\frac{X_{ij}}{x} \Big)\leq 0 $, we thus obtain equation~\eqref{ineq:bound_Hx}. 

Now apply the lower bound of \eqref{ineq:bound_Hx} to $-H(\widehat{X})$
\begin{align*}
    - H(\widehat{X})
    &\leq  \widehat{x} \log(\widehat{x}) - \widehat{x} \\
    &\leq \max\big\{0, \frac{\alpha +\beta}{2} \big[\log(\frac{\alpha+\beta}{2}) -1\big]  \big\}\\
    &\leq \frac{\alpha+\beta}{2} \log(\frac{\alpha+\beta}{2}) - \frac{\alpha+\beta}{2} + 1,
    \end{align*}
where the second inequality is due to $x\log x - x$ being convex and  $0 \le \widehat{x} \leq \frac{1}{2}(\alpha+\beta)$ by Corollary \ref{corollary:solution_sum_bound} and the third inequality is due to $\frac{\alpha +\beta}{2} \big[\log(\frac{\alpha+\beta}{2}) -1\big] + 1 \ge 0$.



Similarly, apply the upper bound of equation \eqref{ineq:bound_Hx} to $H(X^k)$
    \begin{align*}
     H(X^k) &\leq  2x^k \log(n) + x^k - x^k \log(x^k)\\
    &\leq 2x^k \log(n) + 1  \\
    &\leq  \Big(\alpha+\beta + 1 + \frac{1}{2\log(n)}\Big)\log(n) + 1,
    \end{align*}
where the last inequality is due to Lemma~\ref{lemma:bound_xk}(c).
By combining the two results, we have
\begin{align}
    H(X^k) - H(\widehat{X}) \le T,
    \label{bound:hx} 
\end{align}
where $T$ is defined in equation~\eqref{quan:T}.

\paragraph{Upper Bound of $g(X^k) - g(X^*)$.} $ $\\
WLOG we assume that $k$ is odd. At step $k-1$ of Algorithm \ref{Algorithm:Sinkhorn}, we find $u^k$ by minimizing the dual function \eqref{eq:dual_unbalanced_OT} given  $\aA$ and fixed $v^{k-1}$, and simply keep $v^k=v^{k-1}$. Hence, $X^{k}=B(u^{k},v^{k})$ is the optimal solution of
\begin{align*}
\min\limits_{X \in \Br_+^{n \times n}} g^{k}(X) &: = \left\langle C, X\right\rangle - \eta H(X) \\
&+ \tau \KL(X \one_{n} || \aA) + \tau \KL(X^{\top} \one_{n} || \bB^{k}),
\end{align*}
where $\bB^k = \exp\left(\frac{v^k}{\tau} \right) \odot \left[  (X^k)^T \one_n \right]$ with $\odot$ denoting element-wise multiplication. 

Denote $\sum_{i=1}^n \bB_i^k  =  \beta^k$. By Lemma \ref{lemma:key_equation},
\begin{align*}
g^{k}(X^{k}) &= \tau\big(\alpha + \beta^{k}\big) - (2\tau + \eta) x^{k}, \\
g(X^*)&= \tau\big(\alpha + \beta\big) - (2\tau + \eta) x^*.
\end{align*}
Writing $g(X^{k}) - g(X^*) = \left[ g(X^k) - g^k(X^k) \right] + \left[ g^k(X^k) - g(X^*) \right]$, following some derivations using the above equations of $g^{k}(X^{k})$ and $g(X^*)$ and the definitions of $g(X^k)$ and $g^{k}(X^{k})$, we get
{\footnotesize
\begin{align}
g(X^{k}) - g(X^*) &= \left[ -(2\tau +\eta)(x^k -x^*) \right] \nonumber \\
& \quad + \tau \left[ \sum_{j=1}^n b_j^k \log \left( \frac{\bB_j^k}{\bB_j} \right) \right].
\label{eq:diff_g_derivation}
\end{align}
\par}
By part (b) of Lemma \ref{lemma:bound_xk}, the first term is bounded by $\left( 2\tau +\eta \right) \frac{3}{\eta} \Delta_k x^k$.

Note that $\bB_j^k = \exp \left( \frac{v_j^k}{\tau} \right) b_j^k$ and
$\bB_j = \exp \left( \frac{v_j^*}{\tau} \right) b_j^* $. Use part (b) of Lemma \ref{lemma:kth_inequality}, we find that
\begin{align*}
    \log\left(\frac{\bB_j^k}{\bB_j} \right) &= \left[-  \log\left(\frac{b_j^*}{b_j^k}\right) \right] + \frac{1}{\tau} (v_j^k - v_j^*), \\
    \left|\log\left(\frac{\bB_j^k}{\bB_j} \right)\right| 
    &\leq \left( \frac{2}{\eta} \Delta_k \right) + \left ( \frac{1}{\tau} \Delta_k \right) = \left( \frac{2}{\eta} +  \frac{1}{\tau} \right) \Delta_k.
\end{align*}
Note that $b^k_j \ge 0$ for all $j$. The above inequality leads to
\begin{align*}
\left|\sum_{j=1}^n b_j^k \log \left( \frac{\bB_j^k}{\bB_j} \right) \right| &\leq \left( \sum_{j=1}^n b_j^k \right) \max_{1\leq j\leq n} \left|\log \left(\frac{\bB_j^k}{\bB_j}\right) \right| \\
&\leq x^k \left(\frac{2}{\eta} + \frac{1}{\tau}\right) \Delta_k.
\end{align*}
We have
\begin{align*}
    g(X^k) - g(X^*)
    &\leq \left[\left(2\tau + \eta \right) + 3(2\tau +\eta) \right] \left(\frac{\Delta_k}{\eta} \right)  x^k  \nonumber\\
    &\leq 8(\tau + 1)  \left( \frac{\Lambda_{k - 1}}{\eta} \right) S  ,
\end{align*}
where the first inequality is obtained by combining the bounds for two terms of \eqref{eq:diff_g_derivation} while the second inequality results from the fact that $\eta = \frac{\varepsilon}{U} \le \frac{\varepsilon}{2\varepsilon} = \frac{1}{2}$ with $U$ defined in \eqref{quan:U}, part $(c)$ of Lemma \ref{lemma:bound_xk}, and Theorem \ref{theorem:bound_optimal_solution}.

Using part $(a)$ of Lemma \ref{lemma:bound_xk}, this leads to
\begin{align}
    g(X^k) - g(X^*) \le \eta S.
    \label{bound:gx}
\end{align}
Combining \eqref{bound:fx}, \eqref{bound:hx}, \eqref{bound:gx} and the fact that $\eta = \frac{\varepsilon}{U} \le \frac{\varepsilon}{S + T}$, we get
\begin{align*}
    f(X^k) - f(\widehat{X})\leq \eta S + \eta T \leq  \varepsilon.
\end{align*}
As a consequence, we obtain the conclusion of Theorem~\ref{theorem:Sinkhorn_to_optimal}.

\section{Experiments} \label{sec:experiments}
In this section, we provide empirical evidence to illustrate our proven complexity on both synthetic data and real images. In both examples, we vary $\varepsilon$ such that it is small relative to the minimum value of the unregularized UOT function in equation~\eqref{eq:unbalanced_OT} which is computed in advance by using the \textbf{cvxpy} library~\cite{cvxpy_rewriting} with the splitting conic solver option. We then report the two $k$ values: 
\begin{itemize}
    \item The first $k$, denoted by $k_f$, follows the stopping rule in Algorithm \ref{Algorithm:Sinkhorn}.
    \item The second $k$, denoted by $k_c$, is  defined as the minimal $k_c$ such that for all later known  iterations $k'\geq k_c$ in the experiment, Algorithm \ref{Algorithm:Sinkhorn} returns an $\varepsilon$-approximation solution of the UOT problem. 
\end{itemize}  

\subsection{Synthetic Data}
For the simulated example, we choose $n = 100$ and $\tau = 5$. The elements of the cost matrix $C$ are drawn uniformly from the closed interval $[1, 50]$ while those of the marginal vectors $\aA$ and $\bB$ are drawn uniformly from $[0.1, 1]$ and then normalized to have masses $2$ and $4$, respectively. By varying $\varepsilon$ from $1.0$ to $10^{-4}$ (here we uniformly vary $\log(\frac{1}{\varepsilon})$ in the corresponding range for visualization purpose), we follow the scheme presented in the beginning of the section, and report values of $k_f$ and $k_c$ in Figure \ref{fig:k_bound}.

\begin{figure}[h]
    \centering
    \includegraphics[width=8cm]{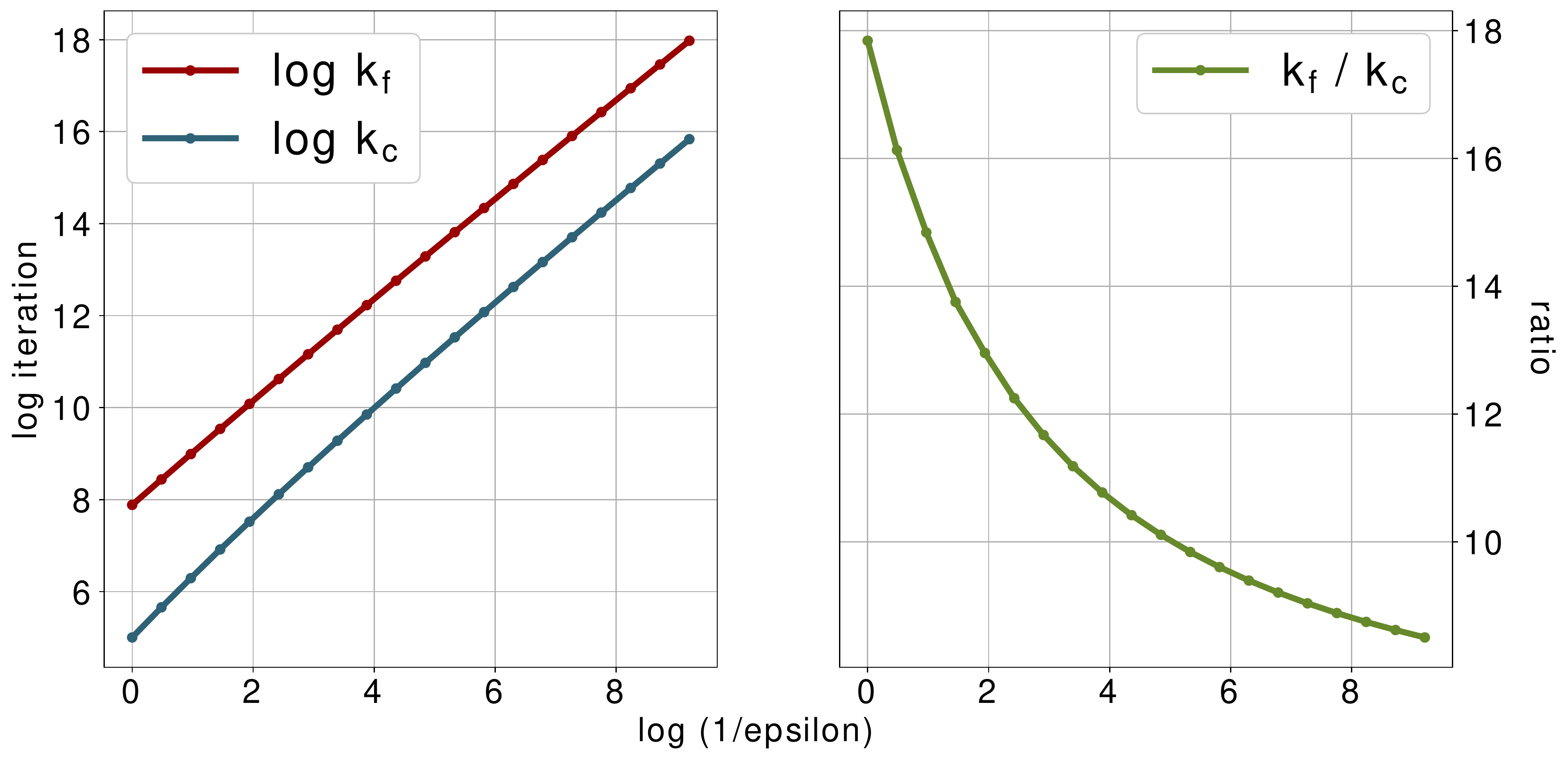}
    \caption{Comparison between $k_c$ and $k_f$ on the synthetic data when varying $\varepsilon$ from $1$ to $10^{-4}$ (and $\eta$ from $3.7(10)^{-2}$ to $3.7(10)^{-6}$ as computed from Theorem~\ref{theorem:Sinkhorn_to_optimal}). The optimal value here is $f(\widehat{X}) = 17.28$, yielding a relative error from $5.8(10)^{-2}$ to $5.8(10)^{-6}$. Both x-axes represent $\log(\frac{1}{\varepsilon}$); the y-axis of the left plot is the natural logarithm of the numbers of iteration (i.e. $k_f, k_c$) and the y-axis of the right plot is the ratio between them.}
    \label{fig:k_bound}
\end{figure}

Figure \ref{fig:k_bound} shows the log values of $k_f, k_c$ stated above when varying $\varepsilon$. When $\varepsilon$ becomes smaller, the left plot indicates that the gap between $k_f$ and $k_c$ becomes narrower, while the right plot shows that the ratio $\frac{k_f}{k_c}$ decreases ( $\frac{k_f}{k_c} \approx 18$ at $\varepsilon = 1$, going down to about $8$ at $\varepsilon = 10^{-4}$). We hypothesize from this trend that our bound becomes more and more accurate as $\varepsilon$ approaches $0$.




\subsection{MNIST Data}
For the MNIST dataset\footnote{http://yann.lecun.com/exdb/mnist/}, we follow similar settings in \cite{Dvurechensky-2018-Computational, altschuler2017near}. In particular, the marginals $\aA, \bB$ are two flattened images in a pair and the cost matrix $C$ is the matrix of $\ell_1$ distances between pixel locations. We also add a small constant $10^{-6}$ to each pixel with intensity 0, except we do not normalize the marginals. We average the results over 10 randomly chosen image pairs and plot the results in Figure~\ref{fig:mnist}. The results on MNIST dataset confirm our theoretical results on the bound of $k$ in Theorem~\ref{theorem:Sinkhorn_to_optimal}. It also shows that  the smaller $\varepsilon$ in the approximation, the closer the empirical result to the theoretical result.

\begin{figure}[!h]
  \centering
  \includegraphics[width=8cm]{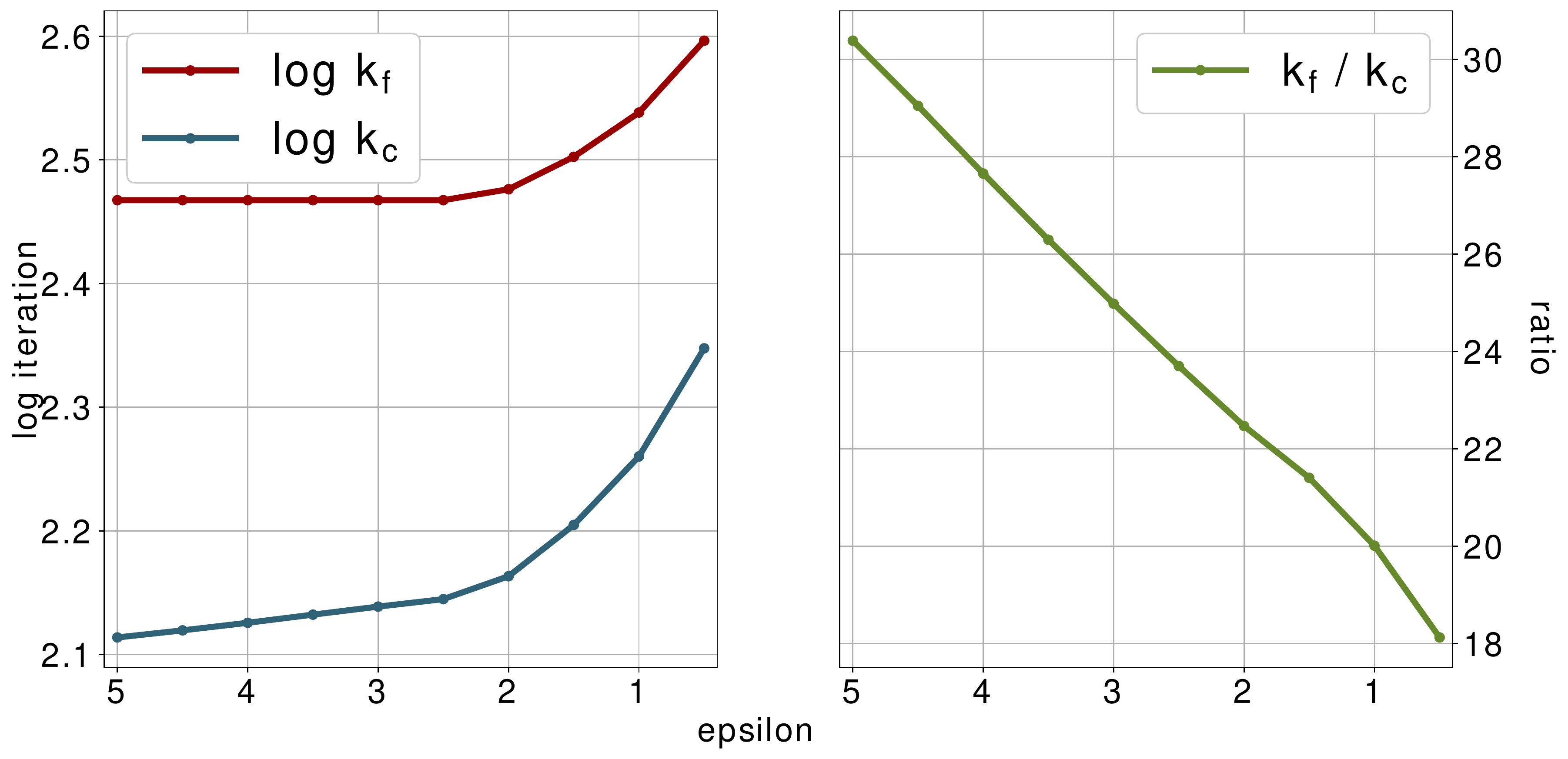}  
  \caption{Comparison between $\log(k_f)$ and $\log(k_c)$ on MNIST when varying $\varepsilon$ from $5$ to $0.5$. We used higher $\varepsilon$ to keep the relative error similar to the first experiment, due to a higher optimal value (among 10 chosen pairs, the minimum was $117.524$ and the maximum was $459.297$.)}
  \label{fig:mnist}
\end{figure}

\subsection{A Further Analysis for Synthetic Data}
In order to investigate  how challenging it is to improve  the theoretical bound for the number of required iterations, we carry out a deeper analysis on the synthetic example. In particular, we set $\eta = 0.5$, $\tau = 5$ and compute the ratios $\frac{\mnorm{v^k - v^*}}{\mnorm{u^{k+1} - u^*}}$ and $\frac{\mnorm{u^{k-1} - u^*}}{\mnorm{v^k - v^*}}$ for even $k$ in range $[0, 100]$ and plot them in Figure \ref{fig:geometric_rate}. As has been proved in Theorem~\ref{theorem:bound_optimal_solution}, these ratios are no less than $\frac{\tau + \eta}{\tau}$. The main reason for this choice is that these differences are used to construct  bounds for many key quantities in lemmas and theorems. These ratios, which are extremely close to $1.1$ for most of the iterations, are consistent with the ratio $\frac{\tau + \eta}{\tau}= 1.1$. Consequently, it is difficult to improve our inequalities in Theorem~\ref{theorem:bound_optimal_solution}.

\begin{figure}[]
    \centering
    \includegraphics[width=8cm]{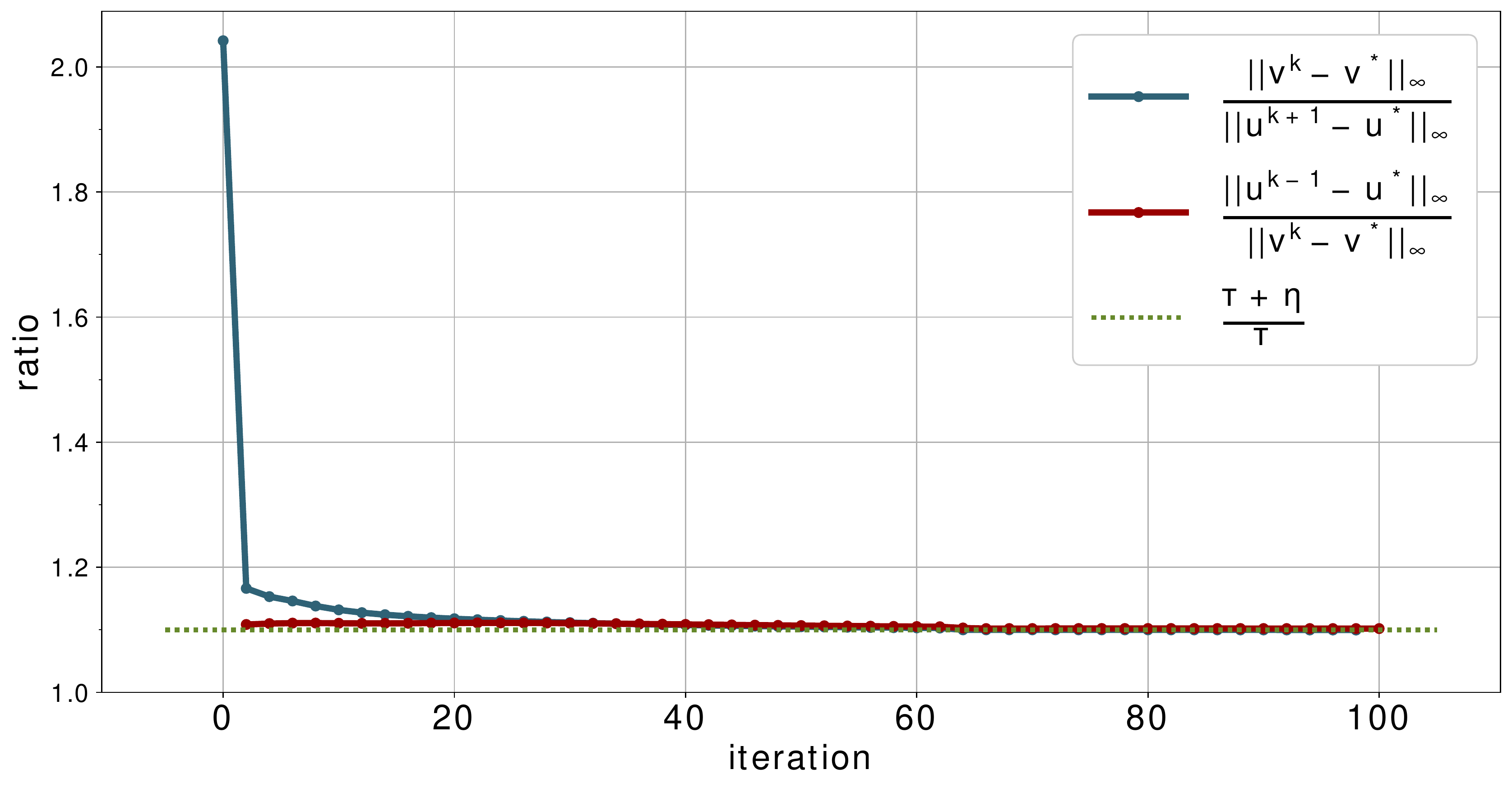}
    \caption{The ratios as observed empirically remain close to our geometric factor for most of the iterations.}
    \label{fig:geometric_rate}
\end{figure}

\section{Discussion}\label{sec:discussions}

In this paper, we prove the near-optimal upper bound of order $\bigOtil(n^{2}/ \varepsilon)$ of the complexity of the Sinkhorn algorithm for approximating the unbalanced optimal transport problem. That complexity is better than the complexity of the Sinkhorn algorithm for approximating the OT problem. In our analysis, some inequalities might not be tight, since we prefer to keep them in simple forms for easier presentation. These sub-optimalities perhaps lead to the inclusion of the logarithmic terms of $\varepsilon$ and $n$ in our complexity upper bound of the Sinkhorn algorithm. We now discuss a few future directions that can serve as natural follow-ups of our work. First, our analysis could be used in the multi-marginal case of UOT by applying Algorithm~\ref{Algorithm:Sinkhorn} repeatedly to every pair of marginals. Second, since the UOT barycenter problem has found several applications in recent years~\cite{Janati_Wasserstein_2019, Schiebinger_Optimal_2019}, it is desirable to establish the complexity analysis of algorithms for approximating it. Finally, similar to the OT problem, the Sinkhorn algorithm for solving UOT also suffers from the curse of dimensionality, namely, when the supports of the measures lie in high dimensional spaces. An important direction is to study efficient dimension reduction scheme with the UOT problem and optimal algorithms for solving it.
\clearpage
\newpage
\bibliography{refs}
\bibliographystyle{icml2020}

\clearpage
\newpage
\onecolumn
\icmltitle{Supplement to "On Unbalanced Optimal Transport: An Analysis of Sinkhorn Algorithm"}
In this appendix, we provide proofs for the remaining results in the paper.
\section{Proofs of Remaining Results}

Before proceeding with the proofs, we state the following simple inequalities:
\begin{lemma}
\label{lemma:basic_inequalities}
The following inequalities are true for all positive $x_i,y_i$, $x$, $y$ and  $0\leq z < \frac{1}{2}$:
\begin{align*}
&(a)\quad  \min_{1\leq i \leq n} \frac{x_i}{y_i}\leq \frac{\sum_{i=1}^n x_i}{\sum_{i=1}^n y_i} \leq \max_{1\leq i \leq n} \frac{x_i}{y_i}, \\
& (b) \quad \exp(z) \leq 1 + |z| + |z|^2, \\
& (c)\quad  \text{ If $\max \Big\{\frac{x}{y}, \frac{y}{x}\Big\} \leq 1 + \delta $, then $|x-y| \leq \delta \min\{x,y \}$},\\
& (d)\quad \left( 1+\frac{1}{x} \right)^{x + 1} \geq e.
\end{align*}
\end{lemma}
\begin{proof}[Proof of Lemma \ref{lemma:basic_inequalities}]
$ $\newline
\textbf{(a)} Given $x_i$ and $y_i$ positive, we have
\begin{align*}
 \min_{1\leq i \leq n} \frac{x_i}{y_i}& \leq \frac{x_j}{y_j} \leq   \max_{1\leq i \leq n} \frac{x_i}{y_i}, \\
 y_j \left( \min_{1\leq i \leq n} \frac{x_i}{y_i} \right) &\leq x_j \leq y_j \left( \max_{1\leq i \leq n} \frac{x_i}{y_i} \right).
\end{align*}
Taking the sum over $j$, we get
\begin{align*}
    \left( \sum_{j=1}^n y_j \right) \left( \min_{1\leq i \leq n} \frac{x_i}{y_i} \right) &\leq \sum_{j=1}^n x_j \leq \left( \sum_{j=1}^n y_j \right) \left( \max_{1\leq i \leq n} \frac{x_i}{y_i} y_j \right), \\
    \min_{1\leq i \leq n} \frac{x_i}{y_i} &\leq \frac{\sum_{j=1}^n x_j}{\sum_{j=1}^n y_j} \leq \max_{1\leq i \leq n} \frac{x_i}{y_i}.
\end{align*}
\textbf{(b)} For the second  inequality,  $\exp(x) \leq 1 + |x| + |x|^2$, we have to deal with the case $x>0$. Since $x\leq \frac{1}{2}$,
\begin{align*}
   \exp(x) &= \sum_{n=1}^{\infty} \frac{x^n}{n!} = 1 + x + x^2 - \frac{x^2}{2} + \sum_{n=3}^{\infty} \frac{x^n}{n!} \leq 1+ x + x^2 - \frac{x^2}{2} + \frac{x^3}{6} \sum_{n=3}^{\infty} x^{n-3}, \\
   &\leq 1+ x + x^2 - \frac{x^2}{2} + \frac{x^3}{6} \frac{1}{1-x} \leq 1+ x + x^2 - \frac{x^2}{2} + \frac{x^3}{3}\leq 1 + x + x^2.
\end{align*}
\textbf{(c)} For the third inequality, WLOG assume $x>y$. Then, we have
\begin{align*}
    \frac{x}{y} \leq 1 + \delta  \Rightarrow x \leq y + y \delta \Rightarrow |x-y| \leq y \delta.
\end{align*}
\textbf{(d)} For the fourth inequality, taking the $\log$ of both sides, it is equivalent to $(x+1) \big[\log(x+1) - \log(x)\big] \geq 1 $. By the mean value theorem, there exists a number $y$ between $x$ and $x+1$ such that  $\log(x+1) - \log(x) = 1/y$, then $(x+1)/y \geq 1$. 
\end{proof}

By the choice of $\eta = \frac{\varepsilon}{U}$ and the definition of $U$, we also have the following conditions on $\eta$:
\begin{equation}
\eta \le \frac{1}{2}; \quad \frac{\eta}{\tau} \le \frac{1}{4\log(n)\max\left\{ 1, \alpha + \beta \right\}}.
\label{ineq:bound_eta}
\end{equation}
Now we come to the proofs of lemmas and the corollary in the main text.
\subsection{Proof of Lemma~\ref{lemma:kth_inequality}}
$ $\newline
\textbf{(a) + (b)}:
From the definitions of $a_i^k$ and $a_i^*$, we have
\begin{equation*}
\log \left( \frac{a_i^*}{a_i^k} \right) = \left( \frac{u_i^* - u_i^k}{\eta}\right) + \log \left( \frac{\sum_{j = 1}^{n} \exp(\frac{v^*_j - C_{ij}}{\eta} ) }{ \sum_{j = 1}^{n} \exp(\frac{v^{k}_j - C_{ij}}{\eta} )}\right).
\end{equation*}
The required inequalities are equivalent to an  upper bound and a lower bound for the second term of the RHS.
Apply part (a) of Lemma  \ref{lemma:basic_inequalities}, we obtain 
\begin{align*}
    \min_{1 \leq j \leq n} \frac{v^*_j - v^{k}_j}{\eta}   \le \log\left( \frac{a_i^*}{a_i^k} \right) -  \frac{u_i^* - u_i^k}{\eta} \leq  \max_{1 \leq j \leq n} \frac{v^*_j - v^{k}_j}{\eta}.
\end{align*}
Part (b) follows similarly. Therefore, we obtain the conclusion of Lemma~\ref{lemma:kth_inequality}.
\subsection{Proof of Corollary \ref{corollary:solution_sum_bound}}
Recall that we have proved in Lemma \ref{lemma:key_equation}:
\begin{align*}
g(X^*) + (2\tau + \eta) x^* = \tau (\alpha + \beta), \\
f(\widehat{X}) + 2\tau\widehat{x} = \tau (\alpha + \beta).
\end{align*}
From the second equality and the fact that $f(\widehat{X}) \ge 0$ (it is easy to see that for $X$ that $X_{ij}\ge 0$, the \textbf{KL} terms and $\langle C, X \rangle$ are all non-negative), we immediately have $\widehat{x} \le \frac{\alpha+\beta}{2}$, proving the second inequality. For the first inequality, we have
$g(X^*) \ge - \eta H(X^*)  \ge - 2\eta x^* \log(n) - \eta x^* + \eta x^* \log(x^*)$. Therefore, we find that
\begin{align*}
- 2\eta x^* \log(n) - \eta x^* + \eta x^* \log(x^*) &\leq \tau(\alpha + \beta) - (2\tau + \eta)x^*,  \\
\eta x^* \log(x^*) + 2\Big(\tau - \eta \log(n)\Big) x^* &\leq \tau(\alpha+\beta).
\end{align*}
It follows from the inequality $z \log(z) \geq z - 1$ for all $z > 0$ that
\begin{align*}
\eta (x^* - 1) + 2( \tau - \eta \log(n)) x^*  &\leq \tau(\alpha+\beta), \\
 x^* (2\tau - 2\eta \log(n) + \eta) &\leq  \tau(\alpha+\beta) + \eta.
 \end{align*}
By inequality \eqref{ineq:bound_eta}, $4\eta \log(n) \le \tau$. Then
 \begin{align*}
 x^* &\leq \frac{\tau(\alpha+\beta) + \eta}{2\tau - 2\eta \log(n) + \eta} \leq \frac{\tau(\alpha+\beta) - (\alpha + \beta) \eta \log(n)}{2\tau - 2\eta \log(n)} + \frac{(\alpha + \beta)\eta \log(n) + \eta }{2\tau - 2\eta \log(n)}, \\
  &\leq \frac{\alpha + \beta}{2}+(\alpha + \beta) \frac{\eta \log(n)}{2\tau -2\eta \log(n)} + \frac{ \frac{\tau}{4\log(n)}}{\frac{3}{2}\tau}\leq \Big(\frac{1}{2}+ \frac{\eta \log(n)}{2\tau - 2 \eta \log(n)}\Big)(\alpha+\beta) + \frac{1}{6\log(n)}.
\end{align*}
As a consequence, we obtain the conclusion of the corollary.

\subsection{Proof of Lemma \ref{lemma:bound_xk}}
$ $ \newline
\textbf{(a)} We prove that $\Lambda_{k} \leq \frac{\eta^2}{8(\tau+1)}$ for $\eta = \frac{\varepsilon}{U}$ and $k \geq \left( \frac{\tau}{\eta} + 1 \right) \left[ \log(8\eta R \big) + \log(\tau(\tau+1)) + 3\log(\frac{1}{\eta}) \right]$ (note that the stated bound can be obtained by replacing $k$ with $k - 1$).

Denote  $\frac{8\eta R (\tau +1)}{\tau^2} = D$ and $\frac{\eta}{\tau}=s > 0$. From inequality \eqref{ineq:bound_eta}, we have $s<1$. The required inequality is equivalent to
\begin{align*}
\frac{\eta^2}{8(\tau+1)} &\geq  \Big(\frac{\tau}{\tau + \eta}\Big)^k \tau R 
\iff \Big(\frac{\tau + \eta}{\tau} \Big)^k \frac{\eta^3}{\tau^3} \geq \frac{8\eta R (\tau +1)}{\tau^2} 
\iff \big(1+s\big)^k s^3 \geq D.
\end{align*}

Let $t = 1 + \frac{\log(D)}{3\log(\frac{1}{s})}$. By definition \eqref{quan:R}, $R \ge \log(n)$, thus $D \ge \frac{8\eta\log(n)(\tau+1)}{\tau^2} > \frac{\eta^3}{\tau^3} = s^3$ and $t >  1 + \frac{3\log(s)}{3\log(\frac{1}{s})} = 0$. We claim the following chain of inequalities 
\begin{align*}
s^3 (1 + s)^k &\ge s^3 (1+s)^{(\frac{1}{s}+1)3\log(\frac{1}{s})t}\\
&\ge s^3 e^{3\log(\frac{1}{s})t}.
\end{align*}
The first inequality results from $k \ge \left(\frac{\tau U}{\epsilon} + 1 \right) \left[ \log(8\eta R \big) + \log(\tau(\tau+1)) + 3\log \left( \frac{U}{\epsilon} \right) \right] = \big(1+\frac{1}{s}\big) 3 \log\big(\frac{1}{s}\big) t > 0$ (using the definitions of $D$, $s$, the choice of $t$ and $\eta = \frac{\varepsilon}{U}$). The second inequality is due to part (d) of Lemma \ref{lemma:basic_inequalities}. The last equality is
\begin{equation*}
s^3 e^{3\log(\frac{1}{s})t} = \frac{1}{s^{3t-3}} = \frac{1}{s^{\log(D)/\log(1/s)}} = \frac{1}{s^{-\log_s(D)}} = D.
\end{equation*}
We have thus proved our claim of part (a).

\textbf{(b)} We need to prove $| x^k - x^*| \le \frac{3}{\eta} \min \{x^*,x^k\} \Delta_k$. From the definition of $x^k$ and $x^*$ and note that they are non-negative:
\begin{align*}
x^k = \sum_{i,j=1}^n \exp\left(\frac{u_i^{k} + v_j^{k} - C_{ij}}{\eta} \right) \quad \text{and} \quad x^* = \sum_{i,j=1}^n \exp\left(\frac{u_i^* + v_j^* - C_{ij}}{\eta} \right).
\end{align*}

Now, we have
\begin{align*}
    \frac{\exp\Big(\frac{u_i^k + v_j^k - C_{ij}}{\eta} \Big)}{\exp\Big(\frac{u_i^* + v_j^* - C_{ij}}{\eta} \Big)} 
    = \exp{\left( \frac{u_i^k - u_i^*}{\eta} \right)} \exp{\left( \frac{v_j^k - v_j^*}{\eta} \right)} 
    \le \left[ \max_{1\le i\le n} \exp\left(\frac{|u_i^{k}- u_i^*|}{\eta} \right) \right] \left[ \max_{1\le j \le n} \exp\left(\frac{|v_j^{k}- v_j^*|}{\eta} \right) \right].
\end{align*}

Note that each of $x^k$ and $x^*$ is the sum of $n^2$ elements and the ratio between  $\exp\Big(\frac{u_i^k + v_j^k - C_{ij}}{\eta} \Big)$ and $\exp\Big(\frac{u_i^* + v_j^* - C_{ij}}{\eta} \Big)$ is bounded by $\left[ \displaystyle \max_{1\le i\le n} \exp\left(\frac{|u_i^{k}- u_i^*|}{\eta} \right) \right] \left[ \displaystyle \max_{1\le j \le n} \exp\left(\frac{|v_j^{k}- v_j^*|}{\eta} \right) \right]$ for all pairs $i,j$. 
Apply part (a) of Lemma \ref{lemma:basic_inequalities}, we find that
\begin{align*}
\max \left\{\frac{x^*}{x^k}, \frac{x^k}{x^*}\right\} 
\leq \left[ \max_{1\le i\le n} \exp\left(\frac{|u_i^{k}- u_i^*|}{\eta} \right) \right] \left[ \max_{1\le j \le n} \exp\left(\frac{|v_j^{k}- v_j^*|}{\eta} \right) \right].
\end{align*}
We have proved from part (a) that $\Lambda_{k -1} \leq \frac{\eta^2}{8(\tau+1)} \le \frac{\eta^2}{8}$. From Theorem  \ref{theorem:bound_optimal_solution} we get $\Delta_k \le \Lambda_{k-1}$. It means that \begin{align*}
    \max_{i,j}\Big\{\Big|\frac{u_i^k - u_i^*}{\eta}\Big|, \Big|\frac{v_j^k -v_j^*}{\eta} \Big| \Big\} = \frac{\Delta_k}{\eta}\leq \frac{\Lambda_{k-1}}{\eta} \leq \frac{\eta}{8} \leq  \frac{1}{8}.
\end{align*}
Apply part (b) of Lemma \ref{lemma:basic_inequalities},
\begin{align*}
\exp\frac{|u_i^{k}- u_i^*|}{\eta} \leq 1 + \frac{|u_i^{k} - u_i^*|}{\eta} + \Big(\frac{|u_i^{k} - u_i^*|}{\eta}\Big)^2, \quad \text{and} \quad \exp\frac{|v_j^{k} - v_j^*|}{\eta} \leq 1 + \frac{|v_j^{k} - v_j^*|}{\eta} + \Big(\frac{|v_j^{k} - v_j^*|}{\eta}\Big)^2.
\end{align*}
Then, we find that
\begin{align*}
\max \left\{\frac{x^*}{x^k}, \frac{x^k}{x^*}\right\} & \leq \left(1 + \frac{1}{\eta} \Delta_k + \frac{1}{\eta^2} \Delta_k^2
\right) \left(1 + \frac{1}{\eta} \Delta_k + \frac{1}{\eta^2}\Delta_k^2 \right)  = 1 + 2\frac{\Delta_k}{\eta}  + 3\frac{\Delta_k^2}{\eta^2}  + 2\frac{\Delta_k^3}{\eta^3} + \frac{\Delta_k^4}{\eta^4}\\
& \leq 1 + \frac{\Delta_k}{\eta}\Big(2 + 3 \frac{\Delta_k}{\eta} + 2 \frac{\Delta_k^2}{\eta^2} +  \frac{\Delta_k^3}{\eta^3} \Big) \leq 1 + \frac{\Delta_k}{\eta} \Big(2 + 3 \frac{1}{8} + 2 \frac{1}{8^2} + \frac{1}{8^3} \Big) \\
&\leq 1 + 3\frac{\Delta_k}{\eta}.
\end{align*}
Apply part (c) of Lemma \ref{lemma:basic_inequalities}, we get 
\begin{align*}
|x^k - x^*| \le \frac{3}{\eta} \Delta_k \min\{x^k, x^*\}.
\end{align*}
Therefore, we obtain the conclusion of part (b).

\textbf{(c)} From Lemma \ref{lemma:bound_xk}(a) and Theorem \ref{theorem:bound_optimal_solution} we have $\frac{\Delta_k}{\eta} \leq \frac{\Lambda_k}{\eta} \leq \frac{\eta}{8}\leq \frac{1}{12}$.  By part (b) of Lemma \ref{lemma:bound_xk}, we have  $x^k  \leq x^* + \frac{3}{\eta}\Delta_k x^* \leq \frac{3}{2}x^*$. Then, we obtain that
\begin{align*}
    x^k &\leq x^* + \frac{3}{\eta}\Delta_k x^* \leq \left[(\alpha + \beta)\Big( \frac{1}{2} + \frac{\eta \log(n)}{2\tau - 2\eta \log(n)}\Big) + \frac{1}{6\log(n)} \right] \big(1 + 3\frac{\Delta_k}{\eta}\big) \\
    &\leq (\alpha + \beta)\Big(\frac{1}{2} + \frac{\eta \log(n)}{2\tau - 2\eta \log(n)} \Big)\big(1 + 3\frac{\Delta_k}{\eta}\big) + \frac{1}{4\log(n)} \\
    & \leq \frac{1}{2}(\alpha+\beta)  + (\alpha + \beta) \frac{3}{2}\frac{\Delta_k}{\eta} + (\alpha+\beta) \frac{\eta \log(n)}{\tau}  + \frac{1}{4\log(n)}\\
    &\leq \frac{1}{2}\big(\alpha+\beta\big) +  \frac{1}{4} + (\alpha + \beta)\frac{3 \eta}{12\tau} + \frac{1}{4\log(n)} \\
    &\leq \frac{1}{2}\big( \alpha + \beta\big) + \frac{1}{2} + \frac{1}{4\log(n)}.
\end{align*}
As a consequence, we reach the conclusion of part (c). 
\end{document}